\documentclass[makeidx]{amsart}
\usepackage[active]{srcltx}
\usepackage{epsfig}
\usepackage{hyperref}
\usepackage{amsmath}
\usepackage{amssymb}
\usepackage{hyperref}
\newtheorem{Proposition}{Proposition}
  \newtheorem{Remark}{Remark}
  \newtheorem{Corollary}[Proposition]{Corollary}
  \newtheorem{Lemma}[Proposition]{Lemma}
  
  \newtheorem{Theorem}{Theorem}

\def\ds{\displaystyle}
 \def\({\left(} \def\){\right)} \makeindex
\author{M. Huang} \address{Mathematics Department\\The Ohio State University\\Columbus, OH 43210} \title[Gamow Vectors]{Gamow Vectors in a Periodically Perturbed Quantum System}
\begin{document}
\begin{abstract}
  We analyze
 the behavior of the wave function
  $\psi(x,t)$ for one dimensional time-dependent Hamiltonian $H=-\partial_x^2\pm2\delta(x)(1+2r\cos\omega t)$
  where $\psi(x,0)$
  is compactly supported.

  We show that $\psi(x,t)$ has a Borel summable expansion containing finitely many terms of the form
  $\sum_{n=-\infty}^{\infty}e^{i^{3/2}\sqrt{-\lambda_{k}+n\omega i}|x|}A_{k,n}e^{-\lambda_{k}t+n\omega it}$, where
  $\lambda_k$ represents the
  associated resonance. This expression defines Gamow vectors
  and resonances in a rigorous and physically relevant way for all frequencies and amplitudes in a time-dependent model.

  For small amplitude ($|r|\ll 1$) there is one resonance for generic initial conditions. We calculate the position of the resonance and discuss its physical meaning as related to multiphoton ionization. We give qualitative theoretical results as well as numerical calculations in the general case.

\end{abstract}
\maketitle
\section{Introduction}

Gamow vectors and resonances, introduced by Gamow to describe $\alpha$-decay
(cf. \cite{gamow}), are very important mathematical tools in the
study of metastable (or quasistable) states in quantum mechanics (cf.
\cite{gold}). The decay states described by Gamow vectors are also
linked to the Fermi-Dirac golden rule (cf. \cite{madrid1}). There
are numerous definitions of resonances and resonant states, using
the scattering matrix, rigged Hilbert spaces, Green's function,
etc. (cf. \cite{madrid1,fpl} and the references therein) These definitions
rely on the time-independent Schr\"{o}dinger equation, though they may be
extended to time-dependent settings in a perturbative regime (cf. \cite{pert, last}).

In a recent paper \cite{mh1}, the author and his collaborator gave a rigorous definition
of Gamow vectors and resonances for compactly supported time-independent
potentials in one dimension, using Borel summation (for a detailed description of Borel summation, see \cite{mh1,CostinBook}). In this paper,
we study the resonances associated to a time-dependent
periodic potential. In our case, the Gamow vector is of the form of
the so-called Floquet ansatz (cf. \cite{floquet}). Our result holds for all
amplitudes and frequencies of the time-dependent field. In the case of
small amplitude or high frequency, we calculate the resonances asymptotically,
and the real part of the resonances measures the ionization
rate. In this sense, our paper extends the results of \cite{frac,qiu}. As we will see, time dependency introduces new subtleties and complex phenomena.

\section{Setting and Main Results\label{sec:main}}

We consider the time-dependent one-dimensional Schr\"{o}dinger equation
\[
i\hbar\dfrac{\partial}{\partial t}\psi(x,t)=-\dfrac{\hbar}{2m}\dfrac{\partial^{2}}{\partial x^{2}}\psi(x,t)+V(x,t)\psi(x,t)\]
where the potential $V(x,t)$ is a delta function potential well or barrier
with a time-periodic perturbation. In this paper, we consider two simple
but illuminating cases:

(1) delta potential well $V(x,t)=-2A\delta(x)(1+2r\cos\omega t)$

(2) delta potential barrier $V(x,t)=2A\delta(x)(1+2r\cos\omega t)$

Here $A>0$ represents the strength of the potential, $r$ represents
the relative amplitude of the perturbation and $\omega$ the frequency.
Without loss of generality we take $r>0,\omega>0$. We further assume the initial wave function $\psi_{0}(x):=\psi(x,0)$
is compactly supported and $C^{2}$ on its support.

We first normalize the equation by changing variables $x\rightarrow\frac{\hbar}{2mA}x,t\rightarrow\frac{\hbar^{2}}{2mA^{2}}t,\omega\rightarrow\frac{2mA^{2}}{\hbar^{2}}\omega$. Note that this is more than using atomic units since we also used the special property of the delta function $\delta(Ax)=\delta(x)/A$.
The equation becomes

\begin{equation}
i\dfrac{\partial}{\partial t}\psi(x,t)=-\dfrac{\partial^{2}}{\partial x^{2}}\psi(x,t)\mp2\delta(x)(1+r\cos\omega t)\psi(x,t)\label{eq:ori}\end{equation}
(where {}``-'' corresponds to the delta potential well and {}``+''
corresponds to the barrier) We shall focus on the delta potential well and analyze in detail the behavior of the wave function as well
as the resonances of the system for all amplitudes and frequencies.
The analysis of the delta potential barrier is very similar and we
will give the results in Section \ref{sec:Fur} without detailed proofs.

\begin{Theorem}\label{thm} Assume the initial wave function $\psi(x,0)$ is compactly supported and $C^{2}$ on its support, then we have for all $t>0$\begin{multline*}
\psi(x,t)=\sum_{k=1}^{K}\sum_{n=-\infty}^{\infty}e^{i^{3/2}\sqrt{-\lambda_{k}+n\omega i}|x|}A_{k,n}e^{-\lambda_{k}t+n\omega it}\\
-\frac{1}{2\pi i}\sum_{n=-\infty}^{\infty}\int_{0}^{e^{i\theta}\infty}e^{i^{3/2}\sqrt{-q+n\omega i}|x|+n\omega it-qt}\varphi_{n}(-q)dq-\frac{1}{2\pi i}\int_{0}^{e^{i\theta}\infty}F(x,-q)e^{-qt}dq\end{multline*}
where $\lambda_{k}+n\omega i$ are resonances of the system ($\mathrm{Re}(\lambda_{k})>0)$,
$\varphi$ a ramified analytic function with square root branch points
at every $n\omega i$ ($n\in \mathbb{Z}$), and $F$ an explicit function with $\sqrt{p}F(p)$
analytic in $\sqrt{p}$. $\theta$ is a small angle chosen to ensure that
no resonance lies on the path of integration.

Moreover, the coefficients $A_{k,n}$ satisfy the recurrence relation
\begin{equation}
\left(\sqrt{-i}\sqrt{i+n\omega i-\lambda_{k}}-1\right)A_{k,n}=rA_{k,n-1}+rA_{k,n+1}\label{eq:akn}\end{equation}
and $\psi(x,t)$ has the Borel summable representation \begin{multline*}
\psi(x,t)=\left(i^{3/2}r\int_{-\infty}^{\infty}\psi_{0}(x)dx\right)t^{-1/2}+\sum_{n=-\infty}^{\infty}\sum_{k=0}^{\infty}C_{n,k}(x)e^{n\omega it}t^{-3/2-k}\\+\sum_{k=1}^{K}\sum_{n=-\infty}^{\infty}e^{i^{3/2}\sqrt{-\lambda_{k}+n\omega i}|x|}A_{k,n}e^{-\lambda_{k}t+n\omega it}\end{multline*}

\end{Theorem}

%\begin{Corollary}\label{cr1} $\psi(x,t)$ has the Borel summed representation \begin{multline*}
%\psi(x,t)=\sum_{n=-\infty}^{\infty}e^{n\omega it}F_{n}(x,t)
%+\sum_{k=1}^{K}\sum_{n=-\infty}^{\infty}e^{i^{3/2}\sqrt{-\lambda_{k}+n\omega i}|x|}A_{k,n}e^{-\lambda_{k}t+n\omega it}\end{multline*}
%where \[
%F_{n}(x,t)\sim\sum_{k=0}^{\infty}C_{n,k}(x)t^{-3/2-k}\:\mathrm{as}\, t\rightarrow\infty\:(n\neq0)\]
%and\[
%F_{0}(x,t)\sim\left(i^{3/2}r\int_{-\infty}^{\infty}\psi_{0}(x)dx\right)t^{-1/2}+\sum_{k=0}^{\infty}C_{0,k}(x)t^{-3/2-k}\:\mathrm{as}\, %t\rightarrow\infty\]
%\end{Corollary}

\begin{Corollary}\label{cr2} For $1\leqslant k\leqslant K$, the Gamow vector term \[
\sum_{n=-\infty}^{\infty}e^{i^{3/2}\sqrt{-\lambda_{k}+n\omega i}|x|}A_{k,n}e^{-\lambda_{k}t+n\omega it}\]
is a generalized eigenvector of the Hamiltonian, in the sense that
it solves \eqref{eq:ori}, but grows exponentially (in a prescribed fashion) for large $|x|$.

\end{Corollary}

\begin{Proposition}
For small $r$ there is only one array of resonances, i.e. $K=1$. The asymptotic position of the array of resonances and a similar result for large $\omega$ are given in Section \ref{sub:small}.

\end{Proposition}

In the above formulas the branch of the square root is chosen to be
the usual one: $\arg(z)\in(-\pi,\pi]$ and $\arg(\sqrt{z})\in(-\frac{\pi}{2},\frac{\pi}{2}]$. We refer to this choice of branch when we use the phrase ``usual (choice of) branch" in this paper.

For small $r$ we calculate asymptotically the position of the resonance, which is related to the ionization rate. For
generic $r$ we will give numerical results showing that the Gamow vector terms exist for some but not all $r$, and we plot the graph of the positions of resonances with different amplitudes (see Section \ref{sec:Fur}).

\begin{Remark} Theorem \ref{thm} and its corollaries generalize to the case where \[
V(x,t)=\mp2A\delta(x)\left(1+2\sum_{k=1}^{K_{0}}(r_{k}\cos k\omega t+s_{k}\sin k\omega t)\right)\]

\end{Remark}

\section{Proof of Main Results}

\subsection{Integral reformulation of the equation}

We first consider the Laplace transform in $t$

\[
\hat{\psi}(x,p)=\int_{0}^{\infty}e^{-pt}\psi(x,t)dt\]

The existence of this Laplace transform (for $\mathrm{Re}(p)>0)$
follows from the existence of a strongly differentiable unitary propagator
(see Theorem X.71, \cite{rns} v.2 pp 290, see also \cite{early},
\cite{mh1} and \cite{qiu}). As we will see, Theorem 1 follows from
analyzing the singularities (poles and branch points) of the analytic continuation of $\hat{\psi}(x,p)$.

Performing this Laplace transform on \eqref{eq:ori}, we obtain

\begin{multline}
ip\hat{\psi}(x,p)-i\psi_{0}(x)=\\-\dfrac{\partial^{2}}{\partial x^{2}}\hat{\psi}(x,p)-2\delta(x)\hat{\psi}(x,p)-2r\delta(x)\hat{\psi}(x,p-i\omega)-2r\delta(x)\hat{\psi}(x,p+i\omega)\label{eq:lap0}\end{multline}

We then rewrite the above ordinary differential equation as an integral
equation by inverting the operator $\dfrac{\partial^{2}}{\partial x^{2}}+ip$. We have

\[
\hat{\psi}(x,p)=\frac{\sqrt{i}e^{-i^{3/2}\sqrt{p}x}}{2\sqrt{p}}\int_{+\infty}^{x}e^{i^{3/2}\sqrt{p}s}g(s)ds-\frac{\sqrt{i}e^{i^{3/2}\sqrt{p}x}}{2\sqrt{p}}\int_{-\infty}^{x}e^{-i^{3/2}\sqrt{p}s}g(s)ds\]
where \[
g(x)=i\psi_{0}(x)-2\delta(x)\hat{\psi}(x,p)-2r\delta(x)\hat{\psi}(x,p-i\omega)-2r\delta(x)\hat{\psi}(x,p+i\omega)\]

Recalling that $\int_{-\infty}^{\infty}\delta(x)f(x)dx=f(0)$, we
simplify the above integral equation and obtain

\begin{multline}
\hat{\psi}(x,p)=\frac{e^{-i^{3/2}\sqrt{p}x}}{2i^{-3/2}\sqrt{p}}\int_{+\infty}^{x}e^{i^{3/2}\sqrt{p}s}\psi_{0}(s)ds-\frac{e^{i^{3/2}\sqrt{p}x}}{2i^{-3/2}\sqrt{p}}\int_{-\infty}^{x}e^{-i^{3/2}\sqrt{p}s}\psi_{0}(s)ds\\
+\frac{\sqrt{i}e^{i^{3/2}\sqrt{p}|x|}}{\sqrt{p}}\left(\hat{\psi}(0,p)+r\hat{\psi}(0,p-i\omega)+r\hat{\psi}(0,p+i\omega)\right)\label{eq:xn0}\end{multline}

Letting $x=0$ we get an equation for $\hat{\psi}(0,p)$

\begin{multline}
\hat{\psi}(0,p)=\frac{i^{3/2}}{2\sqrt{p}}\int_{+\infty}^{0}e^{i^{3/2}\sqrt{p}s}\psi_{0}(s)ds-\frac{i^{3/2}}{2\sqrt{p}}\int_{-\infty}^{0}e^{-i^{3/2}\sqrt{p}s}\psi_{0}(s)ds\\+\frac{\sqrt{i}}{\sqrt{p}}\left(\hat{\psi}(0,p)+r\hat{\psi}(0,p-i\omega)+r\psi(0,p+i\omega)\right)\label{eq:0n0}\end{multline}
which implies

\begin{multline}\label{above}
\frac{\sqrt{i}}{\sqrt{p}}\left(\hat{\psi}(0,p)+r\hat{\psi}(0,p-i\omega)+r\psi(0,p+i\omega)\right)=\\\hat{\psi}(0,p)-\frac{i^{3/2}}{2\sqrt{p}}\int_{+\infty}^{0}e^{i^{3/2}\sqrt{p}s}\psi_{0}(s)ds-\frac{i^{3/2}}{2\sqrt{p}}\int_{-\infty}^{0}e^{-i^{3/2}\sqrt{p}s}\psi_{0}(s)ds\end{multline}

Substituting (\ref{above}) in \eqref{eq:xn0} we get

\begin{multline}\label{xn0}
\hat{\psi}(x,p)=e^{i^{3/2}\sqrt{p}|x|}\hat{\psi}(0,p)+f(x,p)-e^{i^{3/2}\sqrt{p}|x|}f(0,p)\end{multline}
where \[
f(x,p)=\frac{i^{3/2}e^{-i^{3/2}\sqrt{p}x}}{2\sqrt{p}}\int_{+\infty}^{x}e^{i^{3/2}\sqrt{p}s}\psi_{0}(s)ds-\frac{i^{3/2}e^{i^{3/2}\sqrt{p}x}}{2\sqrt{p}}\int_{-\infty}^{x}e^{-i^{3/2}\sqrt{p}s}\psi_{0}(s)ds\]

Equation (\ref{xn0}) indicates that the analytic continuation of $\hat{\psi}(x,p)$,
as well as its singularities, follows naturally from that of $\hat{\psi}(0,p)$,
so it suffices to analyze $\hat{\psi}(0,p)$ using the recurrence
relation \eqref{eq:0n0}. Later we will perform the inverse Laplace
transform on $\hat{\psi}(x,p)$, justified by estimating $\hat{\psi}(0,p)$
and $f(x,p)$ for large $p$. We will then deform the contour of the
Bromwich integral, which yields the expression in Theorem \ref{thm}. It is
worth noting that to deform the contour it suffices to place a branch
cut of the square root in the left half complex plane, while to analyze
the singularities of $\hat{\psi}(0,p)$ we need to consider a larger region in the Riemann surface. Some delicate points of the analysis stems from the complexity of the Riemann surface, since, as we will see, $\hat{\psi}(0,p)$ has infinitely many branch points and there appears to be a barrier of singularities on the non-principal Riemann sheet.

\subsection{Recurrence relation and analyticity of $\hat{\psi}$}

We rewrite the recurrence relation \eqref{eq:0n0} as \[
\left(\sqrt{-i}\sqrt{p}-1\right)\hat{\psi}(0,p)=r\hat{\psi}(0,p-i\omega)+r\hat{\psi}(0,p+i\omega)+\sqrt{-i}\sqrt{p}f(0,p)\]

We will show that $f(0,p)=\ds\frac{\psi_{0}(0)}{p}+O\left(\ds\frac{1}{p^{3/2}}\right)$
as $p\rightarrow\infty$ in any direction in the right half complex
plane (see Section \ref{sub:Proof-of-Theorem}). It is not \textit{a
priori} clear that $\hat{\psi}(0,p)$ has an inverse Laplace transform.
We thus let $\tilde{\psi}(p)=\hat{\psi}(0,p)-f(0,p)$. The recurrence
relation for $\tilde{\psi}$ is

\begin{equation}
\left(\sqrt{-i}\sqrt{p}-1\right)\tilde{\psi}(p)=r\tilde{\psi}(p-i\omega)+r\tilde{\psi}(p+i\omega)+(1+2r)f(0,p)\label{eq:ppp}\end{equation}

It is convenient to write the recurrence relation in a difference
equation form. Denoting $p=i+in\omega+z$ , $y_{n}(z)=\tilde{\psi}(i+in\omega+z)$,
and $f_{n}(z)=(1+2r)f(0,i+in\omega+z)$, we have

\begin{equation}
\left(\sqrt{-i}\sqrt{i+in\omega+z}-1\right)y_{n}(z)=ry_{n-1}(z)+ry_{n+1}(z)+f_{n}(z)\label{eq:rec}\end{equation}

The associated homogeneous equation is of course

\begin{equation}
\left(\sqrt{-i}\sqrt{i+in\omega+z}-1\right)y_{n}(z)=ry_{n-1}(z)+ry_{n+1}(z)\label{eq:hom}\end{equation}

Let $z_{0}$ be a branch point closest to 0, that is, a point on the imaginary axis satisfying $-z_{0}i=\inf_{n}\{|1+n\omega|\}$
(note that $|z_{0}|\leqslant\frac{\omega}{2}$), and let $n_{0}$ be the
corresponding $n$. Since clearly $y_{n}(z)=y_{n+1}(z-i\omega)=y_{n-1}(z+i\omega)$,
it suffices to consider $\mathrm{Im}(z)\in(-\frac{4}{5}\omega,\frac{4}{5}\omega)$
for the usual branch. In general, if we make a branch cut at ($e^{i\theta}$$\infty$,$z_{0}$)
($\cos\theta\neq0$) we consider the strip-shaped region $\{\mathrm{|Im}(z)-\rho\sin\theta|<\frac{4}{5}\omega,\mathrm{\mathrm{Re}(z)=}\rho\cos\theta,\rho\in\mathbb{R}\}$.

To analytically continue $\mathbf{y}:=\{y_{n}\}$, we consider the
Hilbert space $\mathcal{H}$ defined by

\[
||\mathbf{x}||_{\mathcal{H}}^{2}=\sum_{n=-\infty}^{\infty}(1+|n|^{3/2})|x_{n}|^{2}\]
and the operator $\mathcal{C}_{m}:\mathcal{H}\rightarrow\mathcal{H}$

\[
(\mathcal{C}_{m}\mathbf{y})_{n}(z)=\frac{(1+m\sqrt{i})y_{n}(z)+ry_{n-1}(z)+ry_{n+1}(z)}{\left(\sqrt{-i}\sqrt{i+in\omega+z}+m\sqrt{i}\right)}\:~~ (m\in\mathbb{Z}^{+})\]

It is easy to see that $\mathcal{C}_{m}$ is entire in $r$ and analytic in $\sqrt{z-z_{0}}$ in the region $\mathrm{Re}(z)>-m^{2},\mathrm{Im}(z)\in(-\frac{4}{5}\omega,\frac{4}{5}\omega)$.

\begin{Lemma} $\mathcal{C}_{m}$ is a compact operator for any choice of branch.
\end{Lemma}

\begin{proof} For arbitrarily large $N\in\mathbb{N}$, we consider the finite
rank operator $\mathcal{D}_{m,N}:\mathcal{H}\rightarrow\mathcal{H}$
\[
(\mathcal{D}_{m,N}\mathbf{y})_{n}=\begin{cases}
(\mathcal{C}_{m}\mathbf{y})_{n} & |n|<N\\
0 & \mathrm{otherwise}\end{cases}\]

It is easy to check that \[
||\mathcal{C}_{m}-\mathcal{D}_{m,N}||=O(N^{-1/2})\]

Therefore $\mathcal{C}_{m}$, being the limit of finite rank operators
in operator norm, is compact.

\end{proof}

\begin{Lemma}\label{bigp} The equation

\[
\left(\sqrt{-i}\sqrt{i+in\omega+z}-1\right)y_{n}(z)=ry_{n-1}(z)+ry_{n+1}(z)+g_{n}(z)\]
has a unique solution in $\mathcal{H}$ for $|\mathrm{Re}(z)|>(2r+1)^{2}$,
for all $\mathbf{g}\in\mathcal{H}$. In particular, \eqref{eq:rec}
has a unique solution and \eqref{eq:hom} has only the trivial solution $y=0$.
The conclusion holds as well if $z\neq0$ and $r$ is sufficiently
small. Furthermore, for large $|\mathrm{Re}(z)|$ we have $\mathbf{|y|}=O(|\mathrm{Re}(z)|^{-1/2}|\mathbf{g}|$)
where $\mathbf{|x|}:=\sup_{n}|x_{n}|$.
\end{Lemma}
\begin{proof} Note that under the assumptions above, the
norm of the linear operator $\mathcal{S}:\mathcal{H}\rightarrow\mathcal{H}$

\[
(\mathcal{S}\mathbf{y})_{n}(z)=\frac{ry_{n-1}(z)+ry_{n+1}(z)}{\left(\sqrt{-i}\sqrt{i+in\omega+z}-1\right)}\]
 is smaller than 1, since $\left|\sqrt{-i}\sqrt{i+in\omega+z}-1\right|\geqslant|\sqrt{i+in\omega+z}|-1\geqslant\sqrt{|\mathrm{Re}(z)|}-1>2r$.
 We then have $$\mathbf{y}=\frac{(\mathcal{I}-\mathcal{S})^{-1}\mathbf{g}}{\left(\sqrt{-i}\sqrt{i+in\omega+z}-1\right)}$$
\end{proof}
\begin{Proposition}\label{fred} For every $r\in\mathbb{C}$, there are at most finitely
many $z=z_{1},...,z_{l_{r}}$ for which the homogeneous equation \eqref{eq:hom}
has a nonzero solution $\mathbf{y}$ in $\mathcal{H}$. For all other $z$,
there exists a unique solution to \eqref{eq:rec}. The function $\sqrt{z-z_{0}}\mathbf{y}$
is analytic in both $\sqrt{z-z_{0}}$ and $r$, and it can be analytically
continued on the Riemann surface of $\sqrt{i+in\omega+z}$ to $\arg z\in(-3\pi/2,3\pi/2)$. (in other words, one can rotate the branch cut in the left half complex plane) Moreover,
$z_{1},...,z_{l_{r}}$ are either poles ( in $\sqrt{z-z_{0}}$) or removable
singularities of $\mathbf{y}$, and $y_{n}\,(n\neq n_{0})$ is analytic
in $\sqrt{z-z_{0}}$ when $z$ is close to $z_0$.
\end{Proposition}
\begin{proof} We consider the equation

\[
\mathbf{y}^{[m]}=\mathcal{C}_{m}\mathbf{y}^{[m]}+\dfrac{1}{\left(\sqrt{-i}\sqrt{i+in\omega+z}+m\sqrt{i}\right)}\mathbf{f}\]

Since $\mathcal{C}_{m}$ is compact, analytic in both $r$ and $\sqrt{z-z_{0}}$,
and invertible for $|\mathrm{Re}(z)|>(2r+1)^{2}$, it follows from
the analytic Fredholm alternative (see \cite{rns} Vol 1, Theorem
VI.14, pp. 201) that the proposition is true for every $\mathbf{y}^{[m]}$
(note that the solution of the inhomogeneous equation exists for $|\mathrm{Re}(z)|>(2r+1)^{2}$,
thus there can only be finitely many isolated singularities). Uniqueness
of the solution implies $\mathbf{y}^{[m]}=\mathbf{y}^{[m+1]}$ for all
$r\in\mathbb{C},\mathrm{Re}(z)>-m^{2}$. Thus we
naturally define the analytic continuation of the solution to be $\mathbf{y}:=\mathbf{y}^{[m]}$. Analytic continuation on the Riemann surface
follows from the fact that for fixed $r,z$ ($z$ not on the branch
cut) slightly rotating the branch cut does not change the value of
$\sqrt{i+in\omega+z}$ for any $n\in\mathbb{Z}$. Uniqueness of the solution
thus ensures $\mathbf{y}$ also remains unchanged.

Assume $y_{n}(z)\sim b_{n}(z-z_{0})^{-1/2}$ as $z\rightarrow z_{0}$.
It is easy to see from (\ref{eq:rec}) that \begin{multline*}
\left(\sqrt{-i}\sqrt{i+in\omega+z}-1\right)b_{n}=rb_{n-1}+rb_{n+1}\:(n\neq n_{0})\\
\left(\sqrt{-i}\sqrt{i+in_{0}\omega+z_{0}}-1\right)b_{n_{0}}=rb_{n_{0}-1}+rb_{n_{0}+1}-(1/2+r)i^{3/2}\int_{-\infty}^{\infty}\psi_{0}(x)dx\end{multline*}
The unique solution of this recurrence relation is obviously $$b_{n_{0}}=(1/2+r)i^{3/2}\int_{-\infty}^{\infty}\psi_{0}(x)dx,b_{n}=0\,(n\neq n_{0})$$
\end{proof}
\begin{Corollary} For every $r\in\mathbb{C}$, (\ref{eq:ppp}) has a unique
solution $\tilde{\psi}$. $\sqrt{p}\tilde{\psi}$ is meromorphic in
$p$ with square root branches at every $in\omega$ ($n\in\mathbb{Z}$) and poles at $\{p_{k}+in\omega\}$
($k=1,2...l_r,n\in\mathbb{Z}$).
\end{Corollary}
\begin{proof} In order to recover $p=i+in\omega+z$ from the solution to
(\ref{eq:rec}), we only need to show $y_{n}(z)=y_{n\mp1}(z\pm\omega i)$.
To this end, note that by (\ref{eq:rec}) we have

\begin{multline}
\left(\sqrt{-i}\sqrt{i+in\omega+z}-1\right)y_{n\mp1}(z\pm\omega i)\\=ry_{n\mp1-1}(z\pm\omega i)+ry_{n\mp1+1}(z\pm\omega i)+f_{n\mp1}(z\pm\omega i)\end{multline}
which is the same equation as (\ref{eq:rec}) since $f_{n\mp1}(z\pm\omega i)=f_{n}$.

Thus, uniqueness of the solution (Proposition \ref{fred}) implies $y_{n}(z)=y_{n\mp1}(z\pm\omega i)$.
Note that we need to choose the same branch for all $\sqrt{i+in\omega+z}$.
\end{proof}

We conclude this section with a few observations about the positions
of the poles of $\tilde{\psi}$, including the well-known result of
complete ionization (see \cite{frac,early,qiu}).

\begin{Proposition} For $r>0$, $\mathbf{y}$ has no pole on the imaginary
axis or the right half complex plane, with the usual choice of branch.
\end{Proposition}
\begin{proof} In view of Proposition \ref{fred}, we only need to show the homogeneous
equation \eqref{eq:hom} has no nonzero solution in $\mathcal{H}$.
Multiplying \eqref{eq:hom} by $\overline{y_{n}}(z)$ and summing
in $n$ we get \[
\sum_{n=-\infty}^{\infty}\left(\sqrt{-i}\sqrt{i+in\omega+z}-1\right)|y_{n}|^{2}=2r\sum_{n=-\infty}^{\infty}\mathrm{Re}(y_{n-1}\overline{y_{n}})\]
which implies\[
\sum_{n=-\infty}^{\infty}\sqrt{-i}\sqrt{i+in\omega+z}|y_{n}|^{2}\]
must be real.

If $\mathrm{Re}(z)\geqslant0$ then $\mathrm{Im}(\sqrt{-i}\sqrt{i+in\omega+z})\leqslant0$
for all $n$ and $\mathrm{Im}(\sqrt{-i}\sqrt{i+in\omega+z})<0$ for
all $n<-(1+|z|)/\omega$. Thus $y_{n}=0$ for all $n<-(1+|z|)/\omega$
and \eqref{eq:hom} implies $\mathbf{y}=0$.
\end{proof}

%\begin{Remark} The conclusion is not necessarily true with other choices
%of branch, as we will see later.
%\end{Remark}

\begin{Proposition} \label{pole1}For $r>0$, $\mathbf{y}$ has no pole on the imaginary
axis for any choice of branch.
\end{Proposition}
\begin{proof} Similar to the above. Note that $\mathrm{Re}(z)=0$ implies
$\mathrm{Im}(\sqrt{-i}\sqrt{i+in\omega+z})=0$ for all $n>-(1+\mathrm{Im}(z))/\omega$
and $\mathrm{Im}(\sqrt{-i}\sqrt{i+in\omega+z})$ has the same sign
(and nonzero) for all $n<-(1+\mathrm{Im}(z))/\omega$.
\end{proof}
\begin{Proposition}\label{pole2} Solutions of the homogeneous equation \eqref{eq:hom}
exist in negative conjugate pairs, in the sense that if $z_{1}$ is
a pole of $\tilde{\psi}$, then $-\overline{z_{1}}$ is also a pole
(with a different choice of branch, see proof and comments below).
\end{Proposition}
\begin{proof} Simply note that $(-i)^{1/2}\sqrt{i+in\omega+z}=\overline{(-i)^{1/2}\sqrt{i+in\omega-\overline{z}}}$
if we choose the branches in such a way that in the upper half complex
plane the two square roots are the same, while in the lower half plane
they are opposite.
\end{proof}
In view of the above propositions, we will concentrate our study of
resonances on the left half complex plane. The author believes that the imaginary line on the non-principal Riemann surface is a singularity barrier, and the Proposition \ref{pole2} provides a pseudo-analytic continuation across the barrier. We will not discuss the details in this paper.

\subsection{The homogeneous equation}

As we mentioned in the introduction, poles of $\mathbf{y}$ in the
left half complex plane correspond to resonances of the system. According
to Proposition \ref{fred}, finding these poles is essentially the same as
finding solutions to the homogeneous equation \eqref{eq:hom} in $\mathcal{H}$.

\begin{Lemma} \label{www}Assume the nonzero vector $\mathbf{u}=\{u_{n}\}$ satisfies the
homogeneous recurrence relation \eqref{eq:hom}, and that

\[
\sum_{n=0}^{\infty}(1+|n|^{3/2})|u_{n}|^{2}<\infty\]

Assume also that the nonzero vector $\mathbf{v}=\{v_{n}\}$ satisfies \eqref{eq:hom} and

\[
\sum_{n=-\infty}^{0}(1+|n|^{3/2})|v_{n}|^{2}<\infty\]

Then the homogeneous equation \eqref{eq:hom} has a nonzero solution in
$\mathcal{H}$ if and only if the discrete Wronskian $W:=u_{n}v_{n+1}-v_{n}u_{n+1}=0$.
The solution, if it exists, is a constant multiple of $\mathbf{u}$ (or
equivalently $\mathbf{v}$).
\end{Lemma}
\begin{proof} If $r=0$ the lemma is trivial. Assume $r>0$. We first note
that the recurrence relation \eqref{eq:hom} implies

(1) $W$ is independent of $n$.

(2) for any $n$ and any nonzero vector $\mathbf{x}$ satisfying
that recurrence relation, we have $|x_{n}|^{2}+|x_{n+1}|^{2}\neq0,|x_{n}|^{2}+|x_{n+2}|^{2}\neq0\,(n\neq-1).$

Now assume $W=0$. Since $\mathbf{v}$ is nonzero, there exists $m$ for which $v_{m}\neq0$. Thus we
have $u_{m\pm1}=(u_{m}/v_{m})v_{m\pm1}$. Since $\mathbf{u}\neq0$
we must have $u_{m}\neq0$, for otherwise $u_{m\pm1}=u_{m}=0$. If
$v_{m\pm1}=0$ then $u_{m\pm1}=0$, which implies $|u_{m}-(u_{m}/v_{m})v_{m}|^{2}+|u_{m\pm1}-(u_{m}/v_{m})v_{m\pm1}|^{2}=0$,
meaning $\mathbf{u}=(u_{m}/v_{m})\mathbf{v}$. If $v_{m\pm1}\neq0$
then $u_{m\pm1}\neq0$, which inductively implies again $\mathbf{u}=(u_{m}/v_{m})\mathbf{v}$.
Therefore $\mathbf{u}$ solves \eqref{eq:hom} in $\mathcal{H}$.

If $W\neq0$ then clearly $\mathbf{u}$ and $\mathbf{v}$ are the
two linearly independent solutions of the second order difference equation
\eqref{eq:hom}. Furthermore, we have $\liminf_{n<0}|u_{n}|>0$ and
$\liminf_{n>0}|v_{n}|>0$, since $\limsup_{n>0}|u_{n}|<const.|n|^{-3/4}$
and $\limsup_{n<0}|v_{n}|<const.|n|^{-3/4}$ but $u_{n}v_{n+1}-v_{n}u_{n+1}$ is a nonzero constant. Therefore no nonzero
linear combination of $\mathbf{u}$ and $\mathbf{v}$ can be in $\mathcal{H}$.
Since a second order difference equation cannot have any other solution,
there is no nonzero solution of (\ref{eq:hom}) in $\mathcal{H}$.
\end{proof}
We now give a constructive description of $\mathbf{u}$ and $\mathbf{v}$.
For convenience let $h_{n}(z)=\left(\sqrt{-i}\sqrt{i+in\omega+z}-1\right)$.
We choose $n_{1,2}\in\mathbb{\mathbb{Z}}$ so that $|h_{n}|>2|r|$
for all $n\geqslant n_{1}>0$ and $n\leqslant n_{2}<0$. Let $\mathcal{I}$
be the identity operator. We define $\mathcal{H}_{1,2}$ by

\[
||\mathbf{x}||_{1}^{2}=\sum_{n=n_{1}}^{\infty}(1+|n|^{3/2})|x_{n}|^{2}\]

\[
||\mathbf{x}||_{2}^{2}=\sum_{n=-\infty}^{n_{2}}(1+|n|^{3/2})|x_{n}|^{2}\]

\begin{Proposition} There exist $\mathbf{u}$ and $\mathbf{v}$, analytic
in $r$ and ramified analytic in $z$, satisfying the conditions described in
Lemma \ref{www}. Moreover, $\mathbf{u}(z\pm\omega i)=const.\mathbf{u}(z)$ and $\mathbf{v}(z\pm\omega i)=const.\mathbf{v}(z)$.
\end{Proposition}
\begin{proof} Let $\mathcal{T}_{1}:\mathcal{H}_{1}\rightarrow\mathcal{H}_{1}$
\[
(\mathcal{T}_{1}y)_{n}=\begin{cases}
\ds\frac{r}{h_{n}}(y_{n-1}+y_{n+1}) & n>n_{1}\\\\
\ds\frac{r}{h_{n}}y_{n+1} & n=n_{1}\end{cases}\]

Let $\mathbf{a}=(r/h_{n_{1}},0,0...)\in\mathcal{H}_{1}$.

The equation

\[
\mathbf{u}=\mathcal{T}_{1}\mathbf{u}+\mathbf{a}\]

has a unique solution

\[
\mathbf{u}=(\mathcal{I}-\mathcal{T}_{1})^{-1}\mathbf{a}=\mathbf{a}+\mathcal{T}_{1}\mathbf{a}+\mathcal{T}_{2}^{2}\mathbf{a}...\]

since clearly $||\mathcal{T}_{1}||<1$. It is easy to see that $\mathbf{u}$
satisfies

\[
u_{n}=\begin{cases}
\ds\frac{r}{h_{n}}(u_{n-1}+u_{n+1}) & n>n_{1}\\\\
\ds\frac{r}{h_{n}}(u_{n+1}+1) & n=n_{1}\end{cases}\]

Thus the recurrence relation

\[
u_{n}=\frac{h_{n}}{r}u_{n+1}-u_{n+2}\]
extends $\mathbf{u}$ to a solution of the homogeneous equation \eqref{eq:hom}.
In particular $u_{n_{1}-1}=1$. This solution $\mathbf{u}$ is analytic
in $r$ and $z$ (ramified) locally since $\mathcal{T}_{1}$ and $h_{n}$
are analytic in $r$ and $z$ (ramified), and the uniform limit of
analytic functions is analytic. As $r$ or $|\mathrm{Im}(z)|$ increases
we may analytically continue $\mathbf{u}$ by considering some $n_{3}>n_{1}$
so that $|h_{n}|>2|r|$ for all $n\geqslant n_{3}$. Using the same
procedure as we did for $n_{1}$ we get $\mathbf{\tilde{u}}$.
It is easy to see that $\mathbf{u}=u_{n_{3}}\mathbf{\widetilde{u}}$
for they both satisfy the contractive recurrence relation (in the sup norm)

\[
u_{n}=\begin{cases}
\ds\frac{r}{h_{n}}(u_{n-1}+u_{n+1}) & n>n_{3}\\\\
\ds\frac{r}{h_{n}}(u_{n+1}+u_{n_{3}}) & n=n_{3}\end{cases}\]

Note that this implies $u_{n}\neq0$ for large $n$.

The analytic continuation of $\mathbf{u}$ is, up to a scalar multiple,
periodic in $z$. Note that $\mathbf{u}^{\pm}(z)=\mathbf{u}(z\pm\omega i)$
satisfies (for large $n_{3}>n_{1}$)

\[
u_{n}^{\pm}=\begin{cases}
\ds\frac{r}{h_{n\pm1}}(u_{n-1}^{\pm}+u_{n+1}^{\pm}) & n>n_{1}\\\\
\ds\frac{r}{h_{n\pm1}}(u_{n+1}^{\pm}+u_{n_{3}}^{\pm}) & n=n_{3}\end{cases}\]
while $\mathbf{u}$ satisfies

\[
u_{n\pm1}=\begin{cases}
\ds\frac{r}{h_{n\pm1}}(u_{n\pm1-1}+u_{n\pm1+1}) & n>n_{3}\\\\
\ds\frac{r}{h_{n\pm1}}(u_{n\pm1+1}+u_{n_{3}\pm1}) & n=n_{3}\end{cases}\]

Thus $\mathbf{u}(z\pm\omega i)=\dfrac{u_{n_{3}}(z\pm\omega i)}{u_{n_{3}\pm1}(z)}\mathbf{u}(z)$.

The construction of $\mathbf{v}$ is very similar, namely $\mathbf{v}=(\mathcal{I}-\mathcal{T}_{2})^{-1}\mathbf{b}$
where $\mathcal{T}_{2}:\mathcal{H}_{2}\rightarrow\mathcal{H}_{2}$

\[
(\mathcal{T}_{2}y)_{n}\rightarrow\begin{cases}
\ds\frac{r}{h_{n}}(y_{n-1}+y_{n+1}) & n<n_{2}\\\\
\ds\frac{r}{h_{n}}y_{n-1} & n=n_{2}\end{cases}\]
and $\mathbf{b}=(...,0,0,r/h_{n_{2}})$.
\end{proof}
\begin{Proposition}\label{wray} $W$ is analytic in $r$ and ramified analytic in $z$. Moreover,
$W(z)=0$ if and only if $W(z\pm\omega i)=0$.
\end{Proposition}
\begin{proof}  The first part is obvious. The second part follows from the
relation $\mathbf{u}(z\pm\omega i)=\dfrac{u_{n_{3}}(z\pm\omega i)}{u_{n_{3}\pm1}(z)}\mathbf{u}(z)$
(see the proof of the previous proposition) and the fact that $u_{n_{3}}\neq0$.
\end{proof}
\begin{Remark} Another way of constructing $\mathbf{u}$ and $\mathbf{v}$
is by using continued fractions, see \cite{frac}. The continued fraction
expression is slightly simpler in this particular case, but our iteration
method can be easily generalized to trigonometric polynomial potentials
mentioned in section \ref{sec:main}.
\end{Remark}

\subsection{Resonance for small r\label{sub:small}}

We assume $r>0$ and analyze the resonances of the system for small
$r$ (relative to $\omega$) by locating zeros of $W$, in view of
Lemma \ref{www}. Since we will need to consider different branch choices,
we write for convenience $h_{n}(z)=((-i)^{1/2}\sqrt{i+in\omega+z}-1)$
where the power $1/2$ always indicates the usual choice of branch.

\begin{Lemma} For every choice of branch, there exists a constant $c$ so
that when $\omega>c(r+r^{2})$, we have $|h_{n}|>2r$ for all $n\neq0$.
\end{Lemma}
\begin{proof} Recall that for a branch cut at ($e^{i\theta}$$\infty$,$z_{0}$)
($\cos\theta\neq0$), we consider the strip-shaped region $\Omega_{b}:=\{\mathrm{|Im}(z)-\rho\sin\theta|<\frac{4}{5}\omega,\mathrm{\mathrm{Re}(z)=}\rho\cos\theta,\rho\in\mathbb{R}\}$.
It is easy to see that $c_{1}:=\inf_{n\neq0,z\in\Omega_{b}}|\dfrac{z}{\omega}-in|>0$.
Therefore $|h_{n}(z)|=|(-i)^{1/2}\sqrt{i+in\omega+z}-1|=\dfrac{|in\omega+z|}{|\sqrt{i+in\omega+z}+\sqrt{i}|}\geqslant\dfrac{|in\omega+z|}{\sqrt{|in\omega+z|}+2}\geqslant\dfrac{c_{1}\omega}{\sqrt{c_{1}\omega}+2}>2r$
if $\sqrt{c_{1}\omega}>2r+2\sqrt{r}$. Note that $\dfrac{x^{2}}{x+2}$
is an increasing function for $x>0$.
\end{proof}
\begin{Proposition}\label{small} For small $r$, there is a unique nonzero solution of the homogeneous
equation (\ref{eq:hom}) in the left half complex plane with the usual choice of branch. Moreover, the solution satisfies $$z=\left(\frac{2i}{(1+\omega)^{1/2}-1}-\frac{2i}{i^{-1/2}\sqrt{(1-\omega)i}-1}+\sigma(r)\right)r^{2}$$
where $\sigma(r)$ is analytic in $r$ and $\sigma(0)=0$.
\end{Proposition}
\begin{proof} We choose $n_{1}=1,n_{2}=-1$ to construct $\mathbf{u}$ and
$\mathbf{v}$. Thus $u_{0}=v_{0}=1$ and $W=v_{1}-u_{1}$. We calculate
by iterations

\[
u_{1}=\frac{r}{h_{1}}+\frac{r^{3}}{h_{1}^{2}h_{2}}+\frac{r^{5}}{h_{1}^{5}}R_{1}\]

\[
v_{1}=\frac{h_{0}}{r}-v_{-1}=\frac{h_{0}}{r}-\frac{r}{h_{-1}}-\frac{r^{3}}{h_{-1}^{2}h_{-2}}-\frac{r^{5}}{h_{1}^{5}}R_{2}\]

\[
W=\frac{h_{0}}{r}-\frac{r}{h_{-1}}-\frac{r}{h_{1}}-\frac{r^{3}}{h_{1}^{2}h_{2}}-\frac{r^{3}}{h_{-1}^{2}h_{-2}}-\frac{r^{5}}{h_{1}^{5}}R_{1}-\frac{r^{5}}{h_{1}^{5}}R_{2}\]
where $R_{1,2}$ are bounded for $\omega>c(r+r^{2})$. Note that $|h_{0}(z)|=|\sqrt{i+z}-i^{1/2}|\geqslant|z|/2$
and $\left|\dfrac{r}{h_{n}}\right|\leqslant\dfrac{\sqrt{c_{1}\omega}+2}{c_{1}\omega}r$
for all $n\neq0$.

Now, if $\omega$ is fixed and $r$ is small, $W=0$ implies $h_0(z)=O(r^{2})$.
Hence we must have $z=O(r^{2})$. In addition, we need to make
the choice of branch so that $\sqrt{i}$ is in the first quadrant.
Thus we let $z=(a_{0}+\sigma)r^{2}$ where $\sigma=o(1)$, and we
see that

\[
\frac{W}{r}=\left(\frac{a_{0}}{2i}-\frac{1}{(1+\omega)^{1/2}-1}-\frac{i^{1/2}}{\sqrt{(1-\omega)i}-i^{1/2}}\right)(1+o(1))\]

Thus we have $$a_{0}=\frac{2i}{(1+\omega)^{1/2}-1}-\frac{2i}{i^{-1/2}\sqrt{(1-\omega)i}-1}$$

For small $r$, $W$ is clearly analytic in both $r$ and $\sigma$.
Since the value of $W$ depends only on $\bigcup_{n}\{z:|z-in\omega|<2a_{0}r^{2}\}$,
there are exactly two different $W$ with different choices of branch,
namely $W_{1}:\mathrm{Re}(\sqrt{i})>0,\mathrm{Re}(\sqrt{-i})>0$ and
$W_{2}:\mathrm{Re}(\sqrt{i})>0,\mathrm{Re}(\sqrt{-i})<0$. However,
according to Proposition \ref{pole1} and Proposition \ref{pole2}, they are in fact negative
conjugates to each other, and only one will be in the left half complex
plane. We thus take $W=W_{1}$ for its branch is consistent with the usual branch.

It is easy to verify that

\[
\frac{W}{r}|_{r=0,\sigma=0}=0\]

\[
\frac{\partial}{\partial\sigma}\left(\frac{W}{r}\right)|_{r=0,\sigma=0}=-\frac{i}{2}\neq0\]

Therefore it follows from the implicit function theorem that the position
of the zero of $W$ is given by $$z=\left(\frac{2i}{(1+\omega)^{1/2}-1}-\frac{2i}{i^{-1/2}\sqrt{(1-\omega)i}-1}+\sigma(r)\right)r^{2}$$
where $\sigma(r)$ is analytic in $r$ and $\sigma(0)=0$.

$\sigma(r)$ can be found asymptotically by iterating $\sigma(r)-\frac{2iW}{r}$
as in the standard proof of the implicit function theorem.

Since the usual choice of branch is consistent with $W$, the
zero of $W$ is visible.
\end{proof}
\begin{Corollary} For $r$ small and $\omega>1$, the position of the resonance satisfies
$\lambda_{1}\sim-\dfrac{2\sqrt{\omega-1}}{\omega}r^{2}-\dfrac{2\sqrt{\omega+1}}{\omega}r^{2}i$.
\end{Corollary}
\begin{proof} The corollary follows from the expression of $a_{0}$ with the usual choice
of branch. The fact that it is indeed a resonance, i.e. a pole of
$\mathbf{y}$, will be established in the next subsection.
\end{proof}
\begin{Remark} In the case $\omega\gg1+r^{2}$, an analogous analysis shows
that the position of the resonance is given by $\lambda_{1}\sim-\frac{2r^{2}}{\sqrt{\omega}}-\frac{2r^{2}i}{\sqrt{\omega}}$.
\end{Remark}
\begin{Proposition} For small $r$ the poles (in one vertical array) of $\tilde{\psi}$
are simple and the residues are nonzero for generic $\mathbf{f}$.
\end{Proposition}
\begin{proof} We note that the order of the pole of $(\mathcal{I}-\mathcal{C}_{m})^{-1}$
equals the order of the corresponding zero of $\mathcal{I}-\mathcal{C}_{m}$,
which is a constant by the argument principle (see Lemma \ref{arg} below),
since $\mathcal{I}-\mathcal{C}_{m}$ is analytic in $z$. It is easy
to verify that when $r=0$ the zero of $\mathcal{I}-\mathcal{C}_{m}$
is of order one. Thus the poles are simple.

Let $z=G(r)$ be the continuous functions satisfying $W(G(r),r)=0$,
$G(0)=0$. We consider the residue \[
P(r)=\frac{1}{2\pi i}\oint_{|\zeta-G(r)|=\epsilon}y_{0}(\zeta,r)d\zeta\]

Obviously $P(r)$ is analytic in $r$. For generic $\mathbf{f}$,
$P(0)\neq0$ (in which case $\mathbf{y}$ can be found explicitly).
Thus $P(r)\neq0$ for small $r$.
\end{proof}

\subsection{Resonances in general}

Having analyzed the zeros of $W$ for small $r$, we proceed to consider
the case for general $r$, as well as the poles of $\mathbf{y}$.

For convenience we study the region $\Omega_{\theta,\epsilon}:=\{z:\mathrm{Im}(z)\in[\rho\sin\theta+\frac{z_{0}}{2}-\frac{1}{2}\omega+\epsilon,\rho\sin\theta+\frac{z_{0}}{2}+\frac{1}{2}\omega+\epsilon),\mathrm{\mathrm{Re}(z)=}\rho\cos\theta,\rho\in\mathbb{R}\}\bigcap\{z:\mathrm{|Re}(z)|<(2|r|+2)^{2}\}$,
the branch cut being placed at ($e^{i\theta}$$\infty$,$z_{0}$) ($\cos\theta\neq0$).
It is easy to see that there is exactly one zero and one branch point
inside this region for small $r$ (cf. Section \ref{sub:small}). We note that as long as $z$ is
not located on a branch cut, we may rotate the cut slightly without
changing $W$.

\begin{Lemma} \label{finite0} For every $r$, $W$ has finitely many zeros in $\bigcup_{|\cos\theta|>c_{b}>0}\Omega_{\theta,\epsilon}$
where $c_{b}$ is arbitrary.

\end{Lemma}

\begin{proof} By Lemma \ref{bigp}, there is no zero for $\mathrm{|Re}(z)|>(2|r|+1)^{2}$
and the zeros are isolated. Since the Riemann surface of the square
root has only two sheets and the region $\Omega_{\theta,\epsilon}$
is bounded, $W$ can only have finitely many zeros.
\end{proof}
\begin{Lemma}\label{arg} Assume for some $r_{0}$ and arbitrarily small $\epsilon>0$,
with the branch choice $\arg(z)\in(\theta-\epsilon,\theta+2\pi+\epsilon)$
($-2\pi<\theta\leqslant2\pi$,$\cos\theta\neq0$), $W$ has finitely
many zeros in $\Omega_{\theta,\epsilon}$. Then the number of zeros
remains a constant if $r$ is close to $r_{0}$. Furthermore, each zero
moves continuously with respect to $r$.
\end{Lemma}
\begin{proof} The lemma follows from standard complex analysis arguments. Suppose $W(z,r_{0})$ has zeros $z_{1,}z_{2,...}z_{m}$ inside
$\Omega_{0}$ and $\tilde{z}_{m+1},...\tilde{z}_{m+l}$ on $\partial\Omega_{0}$.
Since $z\in\Omega_{\theta,\epsilon}-\Omega_{\theta,0}$ iff $z-i\omega\in\Omega_{\theta,0}-\Omega_{\theta,\epsilon}$,
we let $z_{m+k}=\tilde{z}_{m+k}+i\omega$ $(1\leqslant k\leqslant l)$.
We may choose small $\epsilon>0$ so that $W(z,r_{0})$ has zeros
$z_{1,}z_{2,...}z_{m+l}$ in $\Omega_{\theta,\epsilon}$ for $\arg(z)\in(\theta-\epsilon,\theta+2\pi+\epsilon)$,
and no other zero in $\Omega_{\theta,2\epsilon}$ for $\arg(z)\in(\theta-2\epsilon,\theta+2\pi+2\epsilon)$.
Let $0<\delta<\epsilon$ be small so that there is at most one zero
or branch point inside any circle of radius $2\delta$, and $W(z,r_{0})$
is analytic (with a suitable choice of branch) in $|z-z_{n}|<2\delta$.
Since $W$ is analytic in both $z$ and $r$, it follows from the
argument principle that for $r$ very close to $r_{0}$

\[
M_{n}(r)=\frac{1}{2\pi i}\oint_{|\zeta-z_{n}|=\delta}\frac{\frac{\partial}{\partial\zeta}W(\zeta,r)}{W(\zeta,r)}d\zeta=1\]

Now we consider the compact region $\Omega':=\{z:\arg(z)\in[\theta-\epsilon,\theta+2\pi+\epsilon]\}\bigcap\overline{\Omega_{\theta,\epsilon}}\setminus\bigcup_{n=1}^{m}\{z:|z-z_{n}|<\delta\}$.
Clearly $|W(z,r_{0})|>0$ for all $z\in\Omega'$. Since $W$ is jointly
uniformly continuous in $z$ and $r$, we have $|W(z,r)|>0$ for all $z\in\Omega'$,
$r$ close to $r_{0}$.

Thus the number of zeros is locally a constant and they move continuously
with respect to $r$.
\end{proof}
\begin{Proposition} \label{finite}For every $r$ there are finitely many zeros of $W$
in any strip $\{z:\mathrm{Im}(z)\in[\tilde{z},\tilde{z}+\omega),\mathrm{\mathrm{Re}}(z)\in\mathbb{R}\}$
for all choices of branch within $|\cos\theta|>c_{b}>0$, and the position of each zero
changes continuously with respect to $r$.
\end{Proposition}
\begin{proof} The conclusion follows from Proposition \ref{wray}, Lemma \ref{finite0} and \ref{arg}. Note that we may choose $\theta$ arbitrarily, thus covering
the whole Riemann surface (except for the imaginary lines).
\end{proof}
As we have shown in Proposition \ref{fred} and Lemma \ref{www}, all poles of $\mathbf{y}$
are located where $W=0$. We summarize the results as

\begin{Proposition} \label{akn}For generic $r$ and $\mathbf{f}$, $\mathbf{y}(z,r)$
has finitely many arrays of poles for any choice of branch with
$|\cos\theta|>c_{b}>0$. Their residues $A_{k,n}$ satisfy the recurrence
relation \[
\left((-i)^{1/2}\sqrt{i+n\omega i-\lambda_{1}}-1\right)A_{k,n}=rA_{k,n-1}+rA_{k,n+1}\]
and $\mathbf{A}_{k}\in\mathcal{H}$.
\end{Proposition}
\begin{proof} The first part is simply a rephrasing of previous results (cf. Proposition \ref{finite}).
The recurrence relation for residues follows from the fact that

\[
A_{k,n}=\frac{1}{2\pi i}\oint_{|\zeta-G(r)|=\epsilon}y_{k,n}(\zeta,r)d\zeta\]
satisfies the homogeneous equation (\ref{eq:hom}) since $\mathbf{y}$
satisfies (\ref{eq:rec}) and

\[
\oint_{|\zeta-G(r)|=\epsilon}f_{n}(\zeta,r)d\zeta=0\]

The above expression for $P_{n}$ also implies $\mathbf{A}_{k}\in\mathcal{H}$
since, by H\"{o}lder's inequality

\begin{multline*}
\sum_{n=-\infty}^{\infty}(1+|n|^{3/2})|A_{k,n}|^{2}\leqslant\sum_{n=-\infty}^{\infty}(1+|n|^{3/2})\oint_{|\zeta-G(r)|=\epsilon}|y_{n}(\zeta,r)|^{2}d|\zeta|\\
=\oint_{|\zeta-G(r)|=\epsilon}\sum_{n=-\infty}^{\infty}(1+|n|^{3/2})|y_{n}(\zeta,r)|^{2}d|\zeta|\leqslant\sup_{|\zeta-G_{1}(r)|=\epsilon}||\mathbf{y}(\zeta,r)||^{2}<\infty\end{multline*}
the last inequality following from the continuity of $\mathbf{y}$ (see also Section \ref{sub:Proof-of-Theorem} below).
\end{proof}
%\begin{Remark} It may seem that there will be only one pole for general $r$,
%but this is not true as we will see in Section \ref{sec:Fur}.
%\end{Remark}

\subsection{\label{sub:Proof-of-Theorem}Proof of Theorem 1}

As we have mentioned before, we will take the inverse Laplace transform
of $\hat{\psi}$ and deform the contour, collecting contributions
from the poles in the process. We first provide the necessary estimates.

\begin{Lemma} \label{fp} Assume $\mathrm{supp}\psi_{0}\in[-M,M]$, then $\sqrt{p}f(x,p)$,
where $f(x,p)$ is as defined in Section \ref{sec:main}, is analytic
in $\sqrt{p}$ with a square root branch at zero. Moreover,

\[
f(x,p)=\frac{\psi_{0}(x)}{p}+O(p^{-3/2})+O(p^{-3/2}e^{Mi^{3/2}\sqrt{p}})\]
for large $|p|$.
\end{Lemma}
\begin{proof} By integration by parts we have\begin{multline*}
f(x,p)=\frac{\psi_{0}(x)}{2p}-\frac{e^{-i^{3/2}\sqrt{p}x}}{2p}\int_{+\infty}^{x}e^{i^{3/2}\sqrt{p}s}\psi_{0}'(s)ds\\+\frac{\psi_{0}(x)}{2p}-\frac{e^{i^{3/2}\sqrt{p}x}}{2p}\int_{-\infty}^{x}e^{-i^{3/2}\sqrt{p}s}\psi_{0}'(s)ds\\
=\frac{\psi_{0}(x)}{p}-\frac{\psi_{0}'(x)}{2i^{3/2}p^{3/2}}+\frac{e^{-i^{3/2}\sqrt{p}x}}{2i^{3/2}p^{3/2}}\int_{+\infty}^{x}e^{i^{3/2}\sqrt{p}s}\psi_{0}''(s)ds\\+\frac{\psi_{0}'(x)}{2i^{3/2}p^{3/2}}-\frac{e^{i^{3/2}\sqrt{p}x}}{2i^{3/2}p^{3/2}}\int_{-\infty}^{x}e^{-i^{3/2}\sqrt{p}s}\psi_{0}''(s)ds\\
=\frac{\psi_{0}(x)}{p}+\frac{i^{-3/2}}{2p^{3/2}}\int_{+\infty}^{0}e^{i^{3/2}\sqrt{p}u}\psi_{0}''(u+x)du-\frac{i^{-3/2}}{2p^{3/2}}\int_{-\infty}^{0}e^{-i^{3/2}\sqrt{p}u}\psi_{0}''(u+x)du\end{multline*}

The lemma then follows.
\end{proof}
\begin{Lemma}\label{123} $\tilde{\psi}(p)$ satisfies

(1) For any compact region $\Omega_{1}\in\mathbb{C}$ which does not
contain any pole of $\tilde{\psi}(p)$, we have

\[
\sup_{p\in\Omega_{1}}\sum_{n=-\infty}^{\infty}(1+|n|^{3/2})|\tilde{\psi}(p+n\omega i)|^{2}<\infty\]
In particular,\[
\sup_{p\in\Omega_{1}}\sum_{n=-\infty}^{\infty}|\tilde{\psi}(p+n\omega i)|<\infty\]

(2) For any $c\geqslant0$, $\int_{c-i\infty}^{c+i\infty}\left|\tilde{\psi}(p)\right|dp<\infty$.

(3) For $|\mathrm{Re}(p)|>(2r+1)^{2}$ we have $$\tilde{\psi}(p)=p^{-1/2}O\left(f(0,p)\right)=O\left(p^{-3/2}\right)+O\left(p^{-2}e^{Mi^{3/2}\sqrt{p}}\right)$$

Note that the $p^{-1/2}$ behavior of $\tilde{\psi}(p)$ near the
origin does not affect the nature of these estimates, so we omit further discussions of that special case.
\end{Lemma}
\begin{proof} (1) Recall that $\tilde{\psi}(i+n\omega i+z)=y_{n}(z)$ and that $\mathbf{y}\in\mathcal{H}$, i.e.  \[
||\mathbf{y}||^{2}=\sum_{n=-\infty}^{\infty}(1+|n|^{3/2})|y_{n}|^{2}<\infty\]
Since $\mathbf{y}$ is continuous in $z$ on the Riemann surface of
the square root, so is $||\mathbf{y}||$. Compactness of $\Omega_{1}$
then implies $\sup_{p\in\Omega_{1}}||\mathbf{y}||<\infty$, from which
the first part follows. The second part follows from the Cauchy-Schwarz
inequality

\begin{multline*}
\sum_{n=-\infty}^{\infty}\sup_{p\in\Omega_{1}}|\tilde{\psi}(p+n\omega i)|=\sum_{n=-\infty}^{\infty}(1+|n|^{3/2})^{-1/2}(1+|n|^{3/2})^{1/2}\sup_{p\in\Omega_{1}}|\tilde{\psi}(p+n\omega i)|\\
\leqslant\sum_{n=-\infty}^{\infty}(1+|n|^{3/2})^{-1}\sum_{n=-\infty}^{\infty}(1+|n|^{3/2})\sup_{p\in\Omega_{1}}|\tilde{\psi}(p+n\omega i)|^{2}<\infty\end{multline*}

(2) Note that by Fubini's theorem and Cauchy-Schwarz inequality (cf.
part (1)) we have \begin{multline*}
\int_{c-i\infty}^{c+i\infty}\left|\tilde{\psi}(p)\right|dp=\sum_{n=-\infty}^{\infty}\int_{0}^{1}\left|\tilde{\psi}(c+n\omega i+si)\right|ds\\=\int_{0}^{1}\sum_{n=-\infty}^{\infty}\left|\tilde{\psi}(c+n\omega i+si)\right|ds\\
\leqslant\sup_{p\in[c-si,c+si]}\sum_{n=-\infty}^{\infty}|\tilde{\psi}(p+n\omega i)|<\infty\end{multline*}

(3) The conclusion follows from Lemma \ref{bigp} and Lemma \ref{fp}.
\end{proof}
\begin{Proposition}
$\psi(x,t)=\frac{1}{2\pi i}\int_{C_{1}}e^{i^{3/2}\sqrt{p}|x|+pt}\tilde{\psi}(p)dp+\frac{1}{2\pi i}\int_{C_{2}}e^{pt}f(x,p)dp$,
where the contours $C_{1,2}$ are as shown in Figure \ref{fc1} and \ref{fc2}. In the process
of deforming the first contour, we collect contributions from the
poles and we slightly rotate the branch cut by a small angle $\theta$
if a pole sits on the usual branch cut.
\end{Proposition}
\begin{proof} We first note that $$\sup_{\mathrm{Im}(p)\geqslant0}\left|e^{i^{3/2}\sqrt{p}|x|}\right|=1$$
and $$\sup_{\mathrm{Im}(p)<0,s\in\mathbb{R}}\left|e^{i^{3/2}\sqrt{p+is}|x|}\right|\leqslant\sup_{v\in\mathbb{R}}\left|e^{i^{3/2}\sqrt{\mathrm{-Im}(p)(-1+iv)}|x|}\right|=e^{c_{1}|x|\sqrt{\mathrm{-Im}(p)}}$$
where $c_{1}=\sup_{v\in\mathbb{R}}\mathrm{Re}\left(i^{3/2}\sqrt{(-1+iv)}\right)<\infty$.

Now, by the Bromwich integral formula \begin{multline*}
\psi(x,t)=\frac{1}{2\pi i}\int_{c-i\infty}^{c+i\infty}e^{pt}\hat{\psi}(x,p)dp\\=\frac{1}{2\pi i}\int_{c-i\infty}^{c+i\infty}e^{i^{3/2}\sqrt{p}|x|+pt}\tilde{\psi}(p)dp+\frac{1}{2\pi i}\int_{c-i\infty}^{c+i\infty}e^{pt}f(x,p)dp\end{multline*}

By Lemma \ref{fp} we have \begin{multline*}
\frac{1}{2\pi i}\int_{c-i\infty}^{c+i\infty}e^{pt}f(x,p)dp\\=\frac{\psi_{0}(x)}{2\pi i}\int_{c-i\infty}^{c+i\infty}\frac{e^{pt}}{p}dp+\frac{1}{2\pi i}\int_{c-i\infty}^{c+i\infty}e^{pt}\left(f(x,p)-\frac{\psi_{0}(x)}{p}\right)dp\\
=\psi_{0}(x)+\frac{1}{2\pi i}\int_{C_{2}}e^{pt}\left(f(x,p)-\frac{\psi_{0}(x)}{p}\right)dp=\frac{1}{2\pi i}\int_{C_{2}}e^{pt}f(x,p)dp\end{multline*}

As for the first contour, we only need to show that (along both sides
of the branch cuts)

\[
\sum_{n=-\infty}^{\infty}\int_{0}^{-qe^{i\theta}}e^{i^{3/2}\sqrt{s+n\omega i}|x|+st+n\omega it}\tilde{\psi}(s+n\omega i)ds<\infty\]
\[
\sum_{n=-\infty}^{\infty}\int_{-qe^{i\theta}+n\omega i}^{-qe^{i\theta}+(n+1)\omega i}e^{i^{3/2}\sqrt{p}|x|+pt}\tilde{\psi}(p)dp<\infty\]
and if the resonance is visible with the usual (or slightly rotated) branch cut, then

\[
\sum_{n=-\infty}^{\infty}|A_{k,n}|<\infty\]

The first two estimates follow from Lemma \ref{123}, since

\begin{multline*}
\left|\sum_{n=-\infty}^{\infty}\int_{0}^{-qe^{i\theta}}e^{i^{3/2}\sqrt{s+n\omega i}|x|+st+n\omega it}\tilde{\psi}(s+n\omega i)ds\right|\\\leqslant\left(\sup_{p\in[0,-qe^{i\theta}]}\sum_{n=-\infty}^{\infty}|\tilde{\psi}(p+n\omega i)|\right)\int_{0}^{-qe^{i\theta}}e^{c_{1}|x|\sqrt{|s|}+st}ds<\infty\end{multline*}
and \begin{multline*}
\left|\sum_{n=-\infty}^{\infty}\int_{-qe^{i\theta}+n\omega i}^{-qe^{i\theta}+(n+1)\omega i}e^{i^{3/2}\sqrt{p}|x|+pt}\tilde{\psi}(p)dp\right|\\\leqslant\sum_{n=-\infty}^{\infty}\int_{0}^{\omega i}\left|e^{i^{3/2}\sqrt{-qe^{i\theta}+n\omega i+s}|x|-qe^{i\theta}t}\tilde{\psi}(-qe^{i\theta}+n\omega i+s)\right|ds\\
\leqslant e^{c_{1}|x|\sqrt{|q|}-q\cos\theta t}\left(\sup_{p\in[0,\omega i]}\sum_{n=-\infty}^{\infty}|\tilde{\psi}(p+n\omega i)|\right)<\infty\end{multline*}

The estimates for the resonances follows from proposition \ref{akn} and the Cauchy-Schwarz inequality. Since $\mathbf{A}_{k}\in\mathcal{H}$, we have
\begin{multline*}
\sum_{n=-\infty}^{\infty}|A_{k,n}|=\sum_{n=-\infty}^{\infty}(1+|n|^{3/2})^{-1/2}(1+|n|^{3/2})^{1/2}|A_{k,n}|\\
\leqslant\sum_{n=-\infty}^{\infty}(1+|n|^{3/2})^{-1}\sum_{n=-\infty}^{\infty}(1+|n|^{3/2})|A_{k,n}|^{2}<\infty\end{multline*}

\end{proof}

\begin{figure}[ht!]
\includegraphics[scale=0.5]{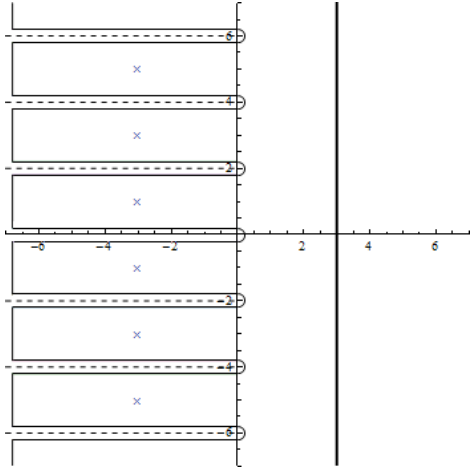}\caption{Contour $C_1$}\label{fc1}

\end{figure}

\begin{figure}[ht!]
\includegraphics[scale=0.5]{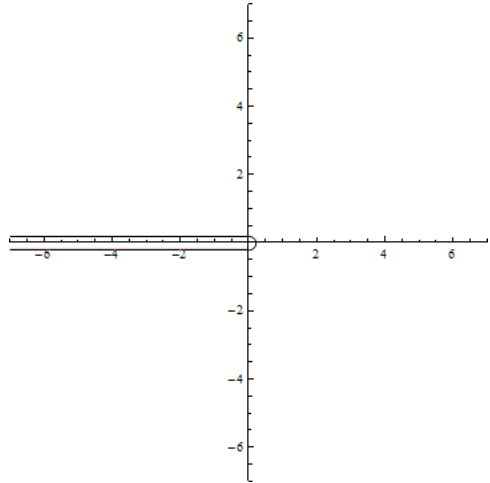}\caption{Contour $C_2$}\label{fc2}

\end{figure}

\begin{Corollary} For $t>0$, we may further deform the contour $C_{1}$
to $C_{3}$ by pushing the vertical lines left to infinity.
\end{Corollary}
\begin{proof} Note that, in the proof of the previous proposition, $\int_{0}^{-qe^{i\theta}}e^{c_{1}|x|\sqrt{|s|}+st}ds$
is bounded in $\mathrm{Re}(q)>0$ and $e^{c_{1}|x|\sqrt{|q|}-q\cos\theta t}\rightarrow0$
as $\mathrm{Re}(q)\rightarrow\infty$.

Thus we conclude the proof of Theorem \ref{thm} by taking the differences
between the upper and lower branches to deform the contour integrals
into line integrals. To be exact, if we denote $F_{s}(x,\sqrt{p})=f(x,p),\tilde{\varphi}_{n}(\sqrt{p-n\omega i})=\hat{\psi}(0,p-n\omega i)$,
then we take $F(x,p)=F_{s}(x,\sqrt{p})-F_{s}(x,-\sqrt{p})$ and $\varphi(p)=\hat{\psi}_{n}(x,\sqrt{p})-\hat{\psi}_{n}(x,\sqrt{p}).$

The last part the theorem follows immediately from Watson's Lemma, since $F$ and
$\varphi$ are clearly analytic in $\sqrt{p}$ and has sub-exponential
growth as $\mathrm{Im}(p)\rightarrow-\infty$ (see Lemma \ref{fp} and \ref{123}).
Note also that $\tilde{\psi}(p)\sim-(1+2r)f(0,p)$ as $p\rightarrow0$.

Corollary \ref{cr2} follows from a direct calculation using \eqref{eq:ori}
and \eqref{eq:akn}. %Note that $\mathrm{Re}(i^{3/2}\sqrt{-\lambda_{k}+n\omega i})>0$.
\end{proof}
\begin{figure}[ht!]
\label{Flo:c3}\includegraphics[scale=0.5]{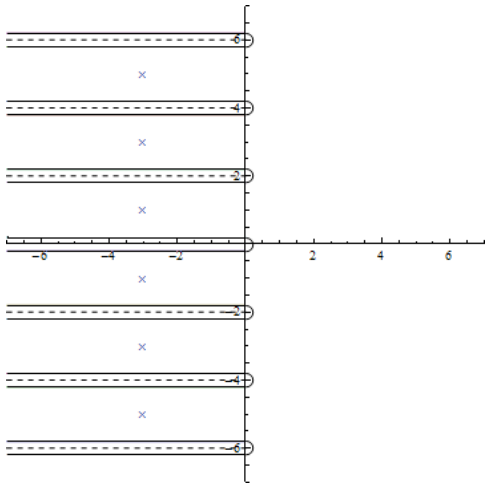}\caption{Contour $C_3$}

\end{figure}

\section{\label{sec:Fur}Further Discussion and Numerical Results}

In this section we study the physical meaning of the resonances, calculate
the positions of the resonances numerically, and discuss the delta potential
barrier.

\subsection{Metastable states and multiphoton ionization}

When a resonance is close to but not on the imaginary axis, it corresponds
to a metastable state of the wave function (see \cite{mh1}). If $|x|$
is not too large, for a moderately long time the wave function is
governed by the Gamow vector terms whose resonances are
closest to the imaginary axis . Thus, for a fixed initial wave function,
the real part of these resonances approximately measure the rate of
ionization, that is, the integral of $|\psi|^{2}$ over a fixed spacial
interval as a function of $t$.

It has been observed (see \cite{frac}) that the rate of ionization
changes rapidly when $\omega$ is approximately equal to an integer
fraction of the bound state energy (in our case, $\omega=1/m,\, m\in\mathbb{N}$).
This phenomenon is related to multiphoton ionization (see \cite{frac,multi,multi1,multi2} and the references therein), a
process in which an electron escapes from the nucleus by absorbing
multiple photons at the same time. Since, as we mentioned in the last
paragraph, the ionization rate can be measured by the position of
resonances, we expect a rapid change in the real part of the resonance
$\lambda_{1}$ when $\omega$ is near $1/m$ and $r$ is small.

\begin{Proposition} For $\frac{1}{m+1}<\omega\leqslant\frac{1}{m}$, the
real part of the resonance is of order $r^{2m+2}$ for small $r$.
\end{Proposition}
{\it Sketch of Proof}: Recalling Proposition \ref{small}, we have $z\sim\frac{2i}{(1+\omega)^{1/2}-1}+\frac{2i}{1-(1-\omega)^{1/2}}$.
It is easy to see that $\mathrm{Im}(h_{n})=O(r^{2})$ for $n\geqslant-1/\omega$.

It can be shown by induction that $(\mathcal{T}_{2}^{k}\mathbf{v})_{1}$
is a function of $h_{-1},h_{-2}...h_{-[\frac{k}{2}]-1}$ and of order
$r^{k+1}$. Moreover, $(\mathcal{T}_{2}^{2k+1}\mathbf{v})_{1}=0$
and $(\mathcal{T}_{2}^{2k}\mathbf{v})_{1}\neq0$.

Therefore, with the notation $z=\left(\frac{2i}{(1+\omega)^{1/2}-1}+\frac{2i}{1-(1-\omega)^{1/2}}+\sigma\right)r^{2}$,
we have $\mathrm{Im}(W)=-\frac{r}{2}\mathrm{Re}(\sigma)(1+o(1))-c_{r}r^{2m+1}(1+O(r))$.
Thus we must have $\mathrm{Re}(\sigma)\sim const.r^{2m}$.

The above proposition implies that there is indeed a rapid
change in the real part of the resonance. Here we confirm this result with
numerical calculations (see Figure \ref{Flo:smallw} below) and omit further details of the proof.

\begin{figure}
\includegraphics[scale=0.5]{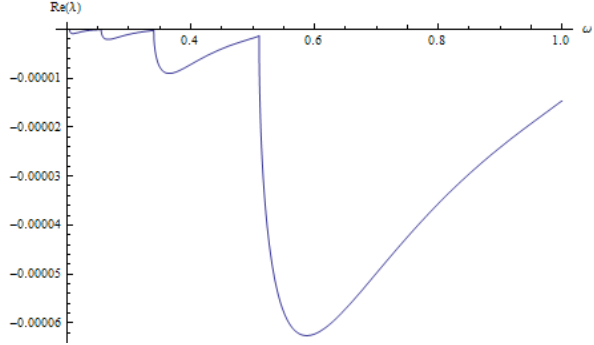}\caption{Real part of the resonance as a function of $\omega$}\label{Flo:smallw}

\end{figure}

\subsection{Position of resonance: numerical results}

As we have shown in Section \ref{sub:small}, for small $r$ there
is only one resonance in the left half complex plane, for all choices of branch. This is, however, not always the case for general
$r$.

We demonstrate the position of resonances in the left half plane by numerically calculating zeros of $W$ for different
$r$. In the graph below we show zeros of $W$ plotted with
different $r$ and choices of branch, with $\omega=2$.

\begin{figure}
\includegraphics[scale=0.5]{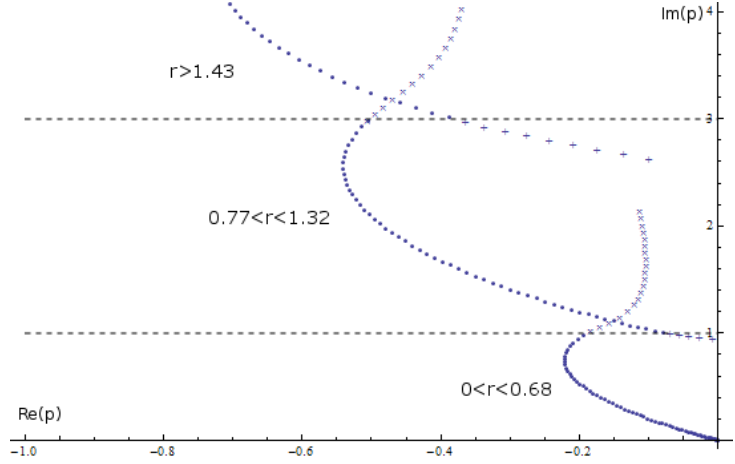}

\caption{Position of resonances for different $r$. Dots are resonances for
the usual branch, and {}``$\times$'' and {}``+'' are those resonances
continuing on the Riemann surface (they are not visible with the usual
branch cut). The {}``$\times$'' and {}``+'' curves in the middle
are on different Riemann sheets.}

\end{figure}

Based on these numerical results, we make the following observations:
\begin{enumerate}
\item For some values of $r$, such as those between 0.69 and 1.31, there
is no visible resonance with the usual choice of branch. In other
words, the Gamow vector term in Theorem 1 is absent.
\item New resonances ({}``+'' marks) emerge as $r$ becomes larger. They
can only be {}``born'' from the imaginary axis, according to Proposition
\ref{finite}.
\item With any given $r$, there does not seem to be more than one resonance
visible with the usual choice of branch.
\item Resonances always move upward with increasing $r$.
\item New resonances move farther away from the imaginary axis compared
to older ones.
\item {}``Old'' resonances ({}``$\times$'' marks) do not move
arbitrarily close to the imaginary axis with increasing $r$.
\end{enumerate}

\subsection{Delta potential barrier}

Finally, we briefly discuss the case for the delta potential barrier.
The corresponding recurrence relation (see (\ref{eq:rec})) is

\[
\left(\sqrt{-i}\sqrt{p}+1\right)\hat{\psi}(0,p)=r\hat{\psi}(0,p-i\omega)+r\hat{\psi}(0,p+i\omega)+\sqrt{-i}\sqrt{p}f(0,p)\]

With a change of branch $\sqrt{p}\rightarrow-\sqrt{p}$ and changes
of variables $r\rightarrow-r,f\rightarrow-f$, the above equation
becomes

\[
\left(\sqrt{-i}\sqrt{p}-1\right)\hat{\psi}(0,p)=r\hat{\psi}(0,p-i\omega)+r\hat{\psi}(0,p+i\omega)+\sqrt{-i}\sqrt{p}f(0,p)\]
which is identical to \eqref{eq:rec}.

Therefore essentially all the theoretical results hold for this case
as well. Note, however, that for small $r$ there is no resonance
with the usual choice of branch (which corresponds to a different
choice of branch in the potential barrier case, see Proposition \ref{small}).

For larger $r$, we expect the behavior of the wave function to be
qualitatively similar to that with a delta potential well, since the
contribution from the time-independent part will be relatively insignificant
compared to the time-dependent part. This is confirmed with the graph
below plotted for different $r$ and $\omega=2$. We choose the usual
branch for simplicity.

\begin{figure}[ht!]
\includegraphics[scale=0.5]{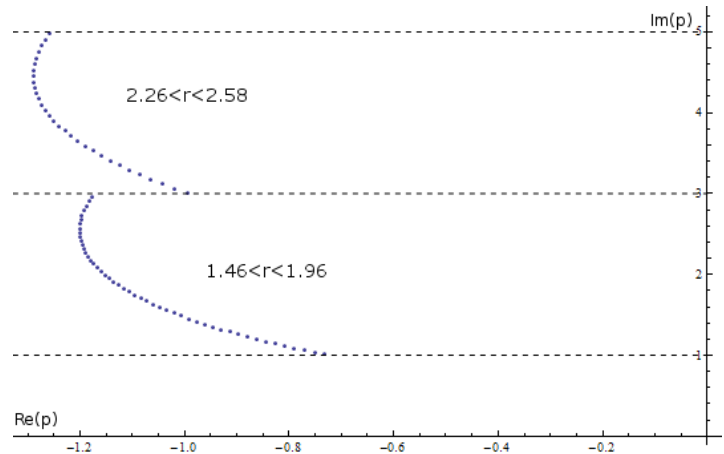}\caption{Position of resonances for different $r$ for the delta potential barrier.}

\end{figure}

\subsection*{Acknowledgments.} The author is grateful to O. Costin who introduced him
to questions of the type addressed in the paper and gave him many valuable suggestions.

\end{document}